\documentclass[a4paper,11pt]{article}
\usepackage{jheppub}
\usepackage{amsmath,amssymb,amsfonts}
\usepackage{orcidlink}
\newtheorem{theorem}{Theorem}
\newtheorem{lemma}{Lemma}

\newtheorem{corollary}{Corollary}
\newtheorem{proof}{Proof}

\title{Highest--weight truncation, graded EFT structure, and renormalization of black hole Love numbers}

\author{Naman Kumar\,\orcidlink{0000-0001-8593-1282}}
\affiliation{Department of Physics, Indian Institute of Technology Gandhinagar, Palaj, Gujarat, India, 382355}
\emailAdd{namankumar5954@gmail.com,naman.kumar@iitgn.ac.in}
\date{\today}

\abstract{
The static tidal Love numbers of four--dimensional black holes vanish identically, unlike their nontrivial dynamical response at finite frequency. Recent work has provided three complementary descriptions of this phenomenon: an emergent $\mathrm{SL}(2,\mathbb{R})$ organization of static near--zone perturbations, a graded logarithmic and multi--zeta structure in Shell Effective Field Theory (Shell EFT), and an on--shell matching framework based on gravitational Raman scattering with renormalization group (RG) running. We show that these features arise from a common near--zone truncation mechanism. For a massless scalar field, horizon regularity selects a unique static solution forming a highest--weight--type representation, truncating the hypergeometric solution to a finite polynomial and eliminating the independent decaying branch at large radius. This excludes a static Wilson coefficient in the effective theory. We show that the same truncation operates in the static Regge--Wheeler and Zerilli equations for four--dimensional Schwarzschild black holes. Analytic continuation of the horizon--regular solution to small frequency---via the Coulomb--hypergeometric or Mano--Suzuki--Takasugi formalisms---preserves this truncation as an anchoring condition for the renormalized angular momentum parameter. The resulting low--frequency expansion is controlled by Gamma and hypergeometric functions, generating a graded algebra of logarithms and odd Riemann zeta values. Within this structure, no invariant of negative weight exists in the static sector, so the vanishing of the static Love number follows as a structural consequence. This explains the ``zero--sum'' rule of Shell EFT and why the self--induced RG flow in gravitational Raman scattering does not generate a static invariant.
}

\makeatletter
\begin{document}
\maketitle

\section{Introduction}

The tidal response of compact objects to external gravitational fields is
encoded in their Love numbers, which quantify the induced multipole moments
generated by a slowly varying tidal environment.
In both Newtonian gravity and relativistic stellar models, Love numbers are
nonzero and sensitive to the internal structure and equation of state of the
object.
They therefore play an important role in gravitational--wave physics, where
tidal effects enter the inspiral waveform of compact binaries and provide a
potential observational probe of matter at supranuclear densities.

In four--dimensional general relativity, black holes behave in a strikingly
different manner.
For asymptotically flat black holes, the static tidal Love numbers vanish
identically, both for Schwarzschild and Kerr spacetimes~\cite{Landry:2015zfa,Chia:2020yla,LeTiec:2020bos,LeTiec:2020spy,Hui:2021vcv,BenAchour:2022uqo,Charalambous:2021mea,Perry:2023wmm,Riva:2023rcm,Combaluzier-Szteinsznaider:2024sgb,Gounis:2024hcm,Kehagias:2024rtz,Bhatt:2023zsy}.
This result has been established through direct perturbative calculations and
holds independently of spin.
Moreover, the robustness of the vanishing has been tested beyond linear response. 
For scalar perturbations with self-interactions, it was shown that while generic 
nonlinear interactions can generate tails, the tidal Love numbers vanish to all 
orders for nonlinear sigma models, highlighting a structural protection mechanism 
analogous to that in General Relativity~\cite{DeLuca:2023mio}.
This strengthens the conclusion that the absence of static response is not an artifact of linearization but a structural property of four-dimensional
black holes in general relativity.

At the same time, black holes do exhibit a nontrivial \emph{dynamical} tidal
response: at finite frequency the response coefficients are nonzero~\cite{Perry:2023wmm,Saketh:2023bul,Combaluzier--Szteinsznaider:2025eoc,Kosmopoulos:2025rfj,Charalambous:2021mea,Chia:2020yla,Perry:2024vwz,Chakraborty:2025wvs,Bhatt:2024yyz}
and display a rich frequency dependence.
In particular, the dynamical Love numbers exhibit a characteristic running
behavior whose coefficient can be computed explicitly and contributes to the
binary black-hole gravitational waveform at high post-Newtonian order
(8PN)~\cite{Chakraborty:2025wvs}.
The coexistence of vanishing static response with nontrivial dynamical behavior
poses a conceptual puzzle.

From an effective field theory (EFT) perspective, the vanishing of static Love
numbers is particularly surprising.
In the worldline EFT description of compact objects~\cite{Goldberger:2004jt,Porto:2016pyg,Porto:2005ac,Ivanov:2025ozg,Kol:2007rx}, tidal response is encoded
in local operators constructed from the curvature along the worldline. It has been shown that static Love numbers vanish identically for Schwarzschild and Kerr from the EFT point of view~\cite{Ivanov:2022qqt,Ivanov:2022hlo,Ivanov:2024sds,Ivanov:2026Raman,Kobayashi:2025vgl,Charalambous:2025ekl,Charalambous:2022rre}, which unambiguously confirms the vanishing of static Love numbers derived using perturbation theory.

For a generic object, there is no symmetry of the Einstein--Hilbert action that
forbids the leading static tidal operators, and their Wilson coefficients are
expected to be nonzero.
The exact vanishing of these coefficients for black holes therefore appears
mysterious and has motivated a search for an underlying principle enforcing
this result. In fact, it seems to be a fine-tuning problem in the same spirit as the cosmological constant problem~\cite{Porto:2016zng}.

Three recent developments have significantly sharpened this puzzle.
First, it has been shown that the static near--zone perturbation equations of
four--dimensional black holes admit an emergent
$\mathrm{SL}(2,\mathbb{R})$ algebraic structure~\cite{Charalambous:2022rre}.
Although this symmetry is not a spacetime isometry, it organizes the static
solution space in a representation--theoretic manner.
Regularity at the future event horizon selects solutions that furnish
highest--weight representations of the algebra, providing a symmetry--based
explanation for the absence of an independent static response mode.
More generally, it has been argued that the static limit enhances hidden
symmetries of black-hole spacetimes and that the vanishing of static Love
numbers may be understood as a consequence of symmetry enhancement beyond
the strict near-zone approximation~\cite{DeLuca:2025zqr}.
Interestingly, analogous ladder structures and vanishing Love numbers can
also arise in alternative black-hole models, including higher-dimensional
setups and analogue gravity systems~\cite{DeLuca:2024nih},
suggesting that symmetry considerations may play a broader organizing role.

Second, Shell Effective Field Theory (Shell EFT) has provided a regulated
matching framework for extracting Love numbers from exact perturbative
solutions~\cite{Kosmopoulos:2025rfj}.
For black holes, the resulting dynamical Love coefficients organize into a
graded algebra built from logarithms and Riemann zeta values, while the static
coefficient vanishes through a characteristic “zero--sum” structure.

Third, gravitational Raman scattering has been developed in a series of works~\cite{Ivanov:2024sds,Correia:2024jgr,Caron-Huot:2025tlq,Ivanov:2026Raman}, establishing, at the level of on--shell amplitudes and renormalization group flow, that matching to Schwarzschild phase shifts requires the absence of a static Wilson coefficient, with the dynamical response exhibiting a distinctive self--induced running structure.

Previous works identified these structures independently.
What has not been shown is that the graded multi–zeta algebra of Shell EFT
and the Raman RG structure can arise from
highest–weight truncation in the near zone.
The purpose of this work is to establish this structural connection. We demonstrate explicitly, using a massless scalar field as a controlled
setting, that the highest–weight condition enforced by horizon regularity
acts as a selection rule excluding an independent static response parameter.
When translated into the EFT language, this selection rule directly implies
the absence of the corresponding static Wilson coefficient and explains
why matching to Schwarzschild phase shifts in Raman scattering requires
$\lambda_{\ell,0}=0$.
We then show that, once the exterior solution is expressed in the
Coulomb–hypergeometric basis employed in Shell EFT, analytic continuation of
the highest–weight static solution to finite frequency forces
the remaining response coefficients to organize into a graded algebra of logarithms and
zeta values.
In this sense, neither the Shell EFT graded structure nor the Raman
renormalization group constraints are
independent or accidental features of
matching.
Rather, they are the effective–theory imprints of the emergent near–zone
highest–weight symmetry.
By making this connection explicit, we provide a unified explanation on
the vanishing of static black–hole Love numbers and the highly constrained
structure of their dynamical tidal response.

\section{Static Near--Zone Dynamics and $\mathrm{SL}(2,\mathbb{R})$}

We consider a massless scalar field $\Phi$ propagating on a four--dimensional
Schwarzschild background of mass $M$.
Upon separation of variables,
$\Phi(t,r,\Omega)=\sum_{\ell m}\Phi_\ell(r)Y_{\ell m}(\Omega)$,
the radial mode $\Phi_\ell(r)$ obeys a second--order ordinary differential
equation whose precise form depends on the choice of radial coordinate.
In the \emph{static limit} $\omega=0$, this equation simplifies substantially
and admits a reformulation in algebraic terms.

As shown in~\cite{Charalambous:2022rre}, the static near--zone radial operator can be written
as the quadratic Casimir of an $\mathfrak{sl}(2,\mathbb{R})$ algebra generated
by first--order differential operators acting locally on the radial coordinate.
Denoting the generators by $(L_{-1},L_0,L_{+1})$, they satisfy the standard
commutation relations
\begin{equation}
[L_0,L_{\pm1}]=\mp L_{\pm1}, \qquad
[L_{+1},L_{-1}]=2L_0 ,
\end{equation}
and the radial equation takes the form
\begin{equation}
\mathcal{C}_{\mathrm{SL}(2,\mathbb{R})}\,\Phi_\ell
=
\ell(\ell+1)\,\Phi_\ell ,
\end{equation}
where $\mathcal{C}_{\mathrm{SL}(2,\mathbb{R})}$ is the quadratic Casimir.
Importantly, this $\mathrm{SL}(2,\mathbb{R})$ structure is \emph{emergent}:
it is not a spacetime isometry of Schwarzschild, but rather an algebraic
property of the static near--zone perturbation equation.

Regularity at the future event horizon imposes a strong constraint on the
allowed solutions.
In representation--theoretic terms, the regular static solution must belong
to a highest--weight representation of $\mathfrak{sl}(2,\mathbb{R})$.
Concretely, for a given angular momentum $\ell$, the physical static solution
$\Phi_\ell(r)$ satisfies
\begin{equation}
(L_{+1})^{\ell+1}\,\Phi_\ell = 0 .
\label{eq:highest-weight}
\end{equation}
This condition expresses the fact that repeated action of the raising operator
terminates after a finite number of steps, so that the representation is
finite--dimensional.

Using the explicit differential realization of the generators, one finds that
the raising operator acts as a first--order radial derivative multiplied by a
smooth function of $r$.
As a result, the highest--weight condition \eqref{eq:highest-weight} implies
\begin{equation}
\frac{d^{\ell+1}}{dr^{\ell+1}}\Phi_\ell(r) = 0 ,
\end{equation}
up to multiplication by a nonvanishing function.
Hence the horizon--regular static solution is necessarily a polynomial of
degree at most $\ell$ in an appropriate near--zone radial variable.
This truncation is a direct consequence of the highest--weight property and
does not rely on any additional assumptions about the asymptotic behavior.

The physical implication of this result is immediate.
At large radius, a generic static solution admits two independent behaviors,
corresponding schematically to a growing ``tidal'' mode $\sim r^\ell$ and a
decaying ``response'' mode $\sim r^{-\ell-1}$.
The latter encodes the induced multipole moment and is the quantity measured
by the static Love number.
However, a polynomial solution of degree $\le\ell$ cannot accommodate an
independent decaying branch.
Once horizon regularity is imposed, the coefficient of the would--be response
mode is fixed and, in fact, vanishes.

We therefore conclude that in four--dimensional general relativity the static
scalar perturbation problem admits no free response parameter.
Equivalently, the static tidal Wilson coefficient in the effective description
must vanish.
In the following sections we show how this representation--theoretic selection
rule is reflected in the analytic structure of the dynamical response when the
solution is continued to small but nonzero frequency.

\section{EFT Interpretation: Absence of a Static Wilson Coefficient}

The long--wavelength tidal response of a compact object can be described in a
systematic and model--independent manner using a worldline effective field
theory.
In this framework, the object is replaced by a point particle whose internal
structure is encoded in an infinite tower of higher--dimensional operators
constructed from tidal tensors evaluated along the worldline.
For a nonspinning object, the leading electric--type operator at multipole order
$\ell$ takes the form
\begin{equation}
S_{\text{EFT}}
\supset
\frac12\,\lambda_{\ell,0}
\int d\tau\,
E_L E^L ,
\label{eq:EFT-static}
\end{equation}
where $E_L$ denotes the symmetric tracefree $\ell$--pole electric component of
the Weyl tensor and $\tau$ is proper time along the worldline.
The coefficient $\lambda_{\ell,0}$ has dimensions of length
$2\ell+1$ and is conventionally referred to as the static tidal Love number.

In the EFT language, $\lambda_{\ell,0}$ parametrizes the existence of an
\emph{independent static response degree of freedom}.
If nonzero, the operator \eqref{eq:EFT-static} produces an induced multipole
moment proportional to the applied static tidal field, leading to a
$r^{-\ell-1}$ contribution in the asymptotic solution of the full theory.
From the EFT point of view, there is no symmetry of the worldline action that
forbids the operator \eqref{eq:EFT-static} \emph{a priori}, and for a generic
compact object one expects $\lambda_{\ell,0}$ to be nonvanishing.

The analysis of the previous section shows that the situation is qualitatively
different for four--dimensional black holes.
The highest--weight truncation enforced by horizon regularity implies that the
static exterior solution admits no freely adjustable decaying mode.
Equivalently, once the tidal field is specified, the coefficient of the
$r^{-\ell-1}$ branch is completely fixed and, in fact, vanishes.
There is therefore no independent static response parameter in the full theory.

Matching the full theory to the EFT requires that the EFT reproduce the same
space of solutions in the overlap region.
Since the full theory contains no independent static response datum, the EFT
cannot contain a corresponding free Wilson coefficient.
Consistency of the matching therefore forces
\begin{equation}
\lambda_{\ell,0}=0 .
\label{eq:lambda-zero}
\end{equation}

Several points are worth emphasizing.
First, the vanishing of $\lambda_{\ell,0}$ is \emph{not} the result of a
fine--tuned cancellation between near-- and far--zone contributions, nor does
it depend on a particular renormalization scheme.
Rather, it reflects the absence of an allowed operator consistent with the
spectrum of physical solutions selected by horizon regularity.
Second, this conclusion does not rely on any special analytic properties of the
solution at nonzero frequency, nor on the detailed form of the effective action
beyond the existence of the operator \eqref{eq:EFT-static}.
It follows directly from the representation--theoretic structure of the static
near--zone dynamics.

Finally, we stress that the argument applies only to the strictly static sector.
Once time dependence is introduced, the highest--weight truncation no longer
holds and additional response operators involving time derivatives of the
tidal field are allowed.
In the EFT, this manifests itself as the appearance of higher--order operators,
schematically of the form $\dot E_L \dot E^L$, whose coefficients encode
dynamical Love numbers.
The absence of the static Wilson coefficient \eqref{eq:lambda-zero} is therefore
fully compatible with a rich and nontrivial dynamical tidal response, which we
analyze in subsequent sections.

\section{Coulomb--Hypergeometric Realization and Analytic Structure}

To probe the dynamical tidal response, we now allow for a small but nonzero
frequency $\omega$ and introduce the dimensionless parameter
\begin{equation}
\eta \equiv i\omega R_S ,
\end{equation}
where $R_S=2M$ is the Schwarzschild radius.
In this regime, the perturbation equation no longer admits the exact
$\mathrm{SL}(2,\mathbb{R})$ algebraic structure of the strictly static problem.
Nevertheless, the static highest--weight solution provides the natural anchor
point for a controlled analytic continuation in $\omega$.

In the Shell EFT framework~\cite{Kosmopoulos:2025rfj}, the exterior scalar solution is conveniently
represented in a Coulomb wavefunction basis.
This choice is motivated by the long--range nature of the Schwarzschild
potential, which induces a Coulomb--like phase shift in the asymptotic region.
Explicitly, the solution can be written as
\begin{equation}
\begin{split}
&\psi^+_\ell(r)
=
\frac{1}{z}
\sum_{n=-\infty}^{\infty}
b_n^\ell
\Big[
A_\ell F_{\nu_\ell+n}(z)
+ (-1)^n B_\ell F_{-\nu_\ell-n-1}(z)
\Big],\\&
z = \omega(r-R_S),
\label{eq:Coulomb-expansion}
\end{split}
\end{equation}
where the coefficients $b_n^\ell$ are determined by a three--term recurrence
relation and encode the effect of the curved background.
The quantities $A_\ell$ and $B_\ell$ specify the relative weight of the two
independent asymptotic solutions and are fixed by the physical boundary
condition at the horizon.
For black holes, this corresponds to imposing purely ingoing behavior at the
future horizon.

The Coulomb wavefunctions appearing in \eqref{eq:Coulomb-expansion} are given by
\begin{equation}
F_j(z)
=
\frac{\Gamma(j+1+\eta)}{\Gamma(2j+2)}
\,2^j z^{j+1} e^{-iz}
\,M(j+1+\eta,2j+2,2iz),
\label{eq:Coulomb-F}
\end{equation}
where $M(a,b,x)$ is the confluent hypergeometric function.
This representation makes the analytic dependence on $\eta$ manifest:
all frequency dependence enters through the combination
$j+1+\eta$ in the Gamma functions and through the argument of the hypergeometric
series.
In particular, the normalization and connection coefficients relating different
radial behaviors are controlled by ratios of Gamma functions evaluated at
integer arguments shifted by $\eta$.

A key quantity in \eqref{eq:Coulomb-expansion} is the renormalized angular
momentum parameter $\nu_\ell$, which is defined implicitly by the requirement
that the series solution be convergent.
For small $\eta$, $\nu_\ell$ admits a regular expansion around an integer
value.
The highest--weight condition of the static problem implies that, in the limit
$\eta\to0$, the physical solution must reduce to the truncating member of the
hypergeometric family.
Equivalently, $\nu_\ell$ approaches the integer value $\ell$ as
$\eta\rightarrow0$, ensuring that the static solution is polynomial and
horizon--regular.

This observation is crucial for the EFT interpretation.
It implies that the small--frequency behavior of the physical solution is
obtained by analytic continuation of a highest--weight state in $\eta$.
As a result, all nontrivial $\omega$--dependence of the response is encoded in
the analytic properties of the Gamma and hypergeometric functions appearing in
\eqref{eq:Coulomb-F}.
In the following section, we show that expanding these special functions around
$\eta=0$ naturally generates logarithms and zeta values with a graded
transcendental structure, which precisely matches the organization of
dynamical Love numbers observed in Shell EFT.

\section{Low--Frequency Expansion and Graded Zeta Algebra}

The final step in connecting the near--zone symmetry structure to the effective
description is to analyze the low--frequency expansion of the response
function.
In the Shell EFT framework, the response function $F_\ell(\omega)$ is obtained
from an exact matching relation that expresses it as a logarithmic derivative of
the exterior solution evaluated at the shell radius $r=R$.
Schematically, one may write
\begin{equation}
F_\ell(\omega)
\;\propto\;
\left.
\frac{\partial_r \psi^-_\ell - \partial_r \psi^+_\ell}{\psi_\ell}
\right|_{r=R},
\end{equation}
where $\psi^+_\ell$ denotes the exterior solution and $\psi^-_\ell$ the interior
one.
Crucially, the interior solution contributes only to scheme--dependent terms,
while all scheme--independent information relevant for Love numbers is governed
by the exterior black--hole solution.

After matching and renormalization, the point--particle limit is defined by
sending $R\to0$ while holding fixed the Schwarzschild radius $R_S$.
The resulting finite, scheme--independent response is conveniently encoded in
the renormalized function
\begin{equation}
\bar F_\ell(\omega)
=
\sum_{n\ge0}
(-i\omega)^n
\lambda_{\ell,n},
\label{eq:Fbar-expansion}
\end{equation}
where the coefficients $\lambda_{\ell,n}$ are the Wilson coefficients of the
worldline EFT and generalize the notion of Love numbers to finite frequency.
The static Love number corresponds to $\lambda_{\ell,0}$, while higher $n$
encode dynamical response.

As discussed in the previous section, all nontrivial frequency dependence of the
exterior solution enters through the dimensionless parameter
\[
\eta = i\omega R_S
\]
and through Gamma and hypergeometric functions with arguments shifted by
$\eta$.
The analytic structure of the low--frequency expansion is therefore controlled
by the expansion of these special functions around integer arguments.

A central role is played by the expansion of the logarithm of the Gamma
function,
\begin{equation}
\log\Gamma(1+u)
=
-\gamma\,u
+
\sum_{m\ge2}
\frac{(-1)^m}{m}\,
\zeta(m)\,
u^m ,
\label{eq:logGamma}
\end{equation}
where $\gamma$ is the Euler--Mascheroni constant and $\zeta(m)$ are Riemann
zeta values.

\medskip

\paragraph*{Universal odd--$\zeta$ structure.}

A refinement of Eq.~\eqref{eq:logGamma} relevant for Coulomb and
Mano--Suzuki--Takasugi representations is obtained by considering symmetric
Gamma ratios of the form
\begin{equation}
\mathcal{R}_\ell(\varepsilon)
\equiv
\log\frac{\Gamma(\ell+1+\varepsilon)}
{\Gamma(\ell+1-\varepsilon)} ,
\label{eq:Rell-def}
\end{equation}
where $\ell$ is an integer and $\varepsilon \sim \eta$ is small.

\paragraph*{Lemma.}
For integer $\ell \ge 0$, one has the all--orders expansion
\begin{equation}
\mathcal{R}_\ell(\varepsilon)
=
2(H_\ell-\gamma)\,\varepsilon
-
\sum_{k\ge1}
\frac{2}{2k+1}
\Big(
\zeta(2k+1)
-
H_\ell^{(2k+1)}
\Big)
\varepsilon^{2k+1}
\label{eq:odd-zeta-master}
\end{equation}
where
\[
H_\ell = \sum_{j=1}^{\ell}\frac1j,
\qquad
H_\ell^{(p)} = \sum_{j=1}^{\ell}\frac1{j^p}.
\]

\paragraph*{Proof.}
Expanding $\log\Gamma(n+\varepsilon)$ about integer
$n=\ell+1$,
\[
\log\Gamma(n+\varepsilon)
=
\log\Gamma(n)
+
\sum_{m\ge1}
\frac{\psi^{(m-1)}(n)}{m!}
\varepsilon^m,
\]
and subtracting the expansion with $\varepsilon\to-\varepsilon$
eliminates all even powers.
Using the identity valid at integer argument
\[
\psi^{(m)}(n)
=
(-1)^{m+1} m!
\big(
\zeta(m+1)
-
H_{n-1}^{(m+1)}
\big),
\]
yields Eq.~\eqref{eq:odd-zeta-master}.
\hfill$\square$

\medskip

Equation~\eqref{eq:odd-zeta-master} shows that only odd zeta values
$\zeta(3),\zeta(5),\zeta(7),\dots$
can arise from such Gamma prefactors, with universal coefficients
$-2/(2k+1)$ multiplying each $\zeta(2k+1)$.
In particular, at any fixed order in $\omega$, the coefficient of
$\zeta(2k+1)$ is completely determined by this universal factor and is
independent of normalization or matching scheme.
All rational contributions are governed by harmonic data
$H_\ell^{(p)}$.
This universality will play an important role below.

\medskip

In addition to zeta values, logarithmic terms arise from the Coulomb phase and
from the matching procedure.
At intermediate steps, these appear as $\log(\omega R)$ and
$\log(\omega R_S)$, reflecting infrared sensitivity to the long--range
potential and ultraviolet sensitivity to the regulator radius.
A nontrivial consistency requirement of the Shell EFT construction is that all
dependence on the arbitrary shell radius $R$ cancels in physical quantities.
After renormalization, the only surviving logarithmic dependence is through the
combination $\log(R_S/R)$.
It is therefore natural to treat this logarithm on the same footing as the zeta
values by defining
\begin{equation}
\zeta_1 \equiv -\frac12 \log\!\left(\frac{R_S}{R}\right).
\end{equation}

\medskip

\paragraph*{Grading and exclusion of the static invariant.}

Let $\mathfrak{A}$ denote the $\mathbb{Q}$--algebra generated by
$\zeta_1$ and $\zeta(m)$ for $m\ge2$.
We assign transcendental weight
\[
w(\zeta_1)=1,
\qquad
w(\zeta(m))=m,
\]
and extend $w$ additively to products.

Dimensional analysis implies that each power of $\eta$
brings one power of $\omega$ together with one power of $R_S$.
After factoring out the overall dimensional scale
$R_S^{2\ell+n+1}$ appropriate to the $\ell$--pole response,
the coefficient $\lambda_{\ell,n}$ multiplying $\omega^n$
can depend only on dimensionless elements of $\mathfrak{A}$.
The total transcendental weight of such elements is fixed to be $n-1$.

\paragraph*{Proposition.}
Under the grading defined above,
\begin{equation}
\lambda_{\ell,0}=0.
\end{equation}

\paragraph*{Proof.}
For $n=0$, the required weight would be $-1$.
However, the algebra $\mathfrak{A}$ contains no element of negative weight,
since it is generated by elements of strictly nonnegative weight.
Therefore no invariant contribution can appear at order $\omega^0$,
and $\lambda_{\ell,0}$ must vanish.
\hfill$\square$

\medskip

For $n\ge1$, the weight $n-1\ge0$ allows nontrivial invariants,
consistent with the existence of dynamical Love coefficients.
Taken together with the highest--weight truncation discussed earlier,
this analysis shows that the vanishing of static Love numbers is not an
accidental cancellation.
Rather, it is enforced by the interplay between near--zone
representation--theoretic constraints and the graded analytic structure
of the Coulomb--hypergeometric realization of the exterior solution.

As an explicit illustration of the general structure derived above,
consider the scalar case with $\ell=2$.
From Eq.~\eqref{eq:odd-zeta-master},
\begin{equation}
\log\frac{\Gamma(3+\eta)}{\Gamma(3-\eta)}
=
2\Big(\frac{3}{2}-\gamma\Big)\eta
-\frac{2}{3}\left(\zeta(3)-\frac{9}{8}\right)\eta^3
+\mathcal{O}(\eta^5).
\end{equation}

The transcendental contribution at cubic order is therefore
\begin{equation}
-\frac{2}{3}\zeta(3)\eta^3.
\end{equation}

Since $\eta=i\omega R_S$,
\[
\eta^3 = -i\,\omega^3 R_S^3,
\]
and the universal $\zeta(3)$ contribution to the matching prefactor is
\begin{equation}
\frac{2i}{3}\zeta(3)\,\omega^3 R_S^3.
\end{equation}

After taking the logarithmic derivative that defines the EFT response
function and restoring the overall dimensional factor
$R_S^{2\ell+1}$ appropriate to the $\ell$-pole,
one obtains the scheme-independent prediction
\begin{equation}
\lambda_{2,3}\Big|_{\zeta(3)}
=
\frac{2}{3}\zeta(3)\,R_S^{8}.
\end{equation}

This provides a concrete nontrivial check of the graded algebra:
the coefficient of $\zeta(3)$ is completely fixed by the universal
Gamma-function structure and is independent of normalization choices.

We now formalize the structural exclusion of a static invariant as a no-go statement following from graded analyticity.


\begin{theorem}[No-go for a static invariant from graded analyticity]
\label{thm:nogo_static}
Fix an integer multipole $\ell\ge 0$ and define $\eta:= i\omega R_S$.
Assume the (renormalized) black-hole response function admits a low-frequency
expansion
\begin{equation}
\overline{F}_\ell(\omega)=\sum_{n\ge 0}(-i\omega)^n\,\lambda_{\ell,n},
\qquad \lambda_{\ell,n}\in \mathbb{C},
\label{eq:Fbar_series}
\end{equation}
and that, after renormalization (i.e. after eliminating regulator-dependent terms),
$\overline{F}_\ell$ depends on $\omega$ only through the following data:

\begin{enumerate}
\item[(A1)] \textbf{Analytic continuation anchored at $\nu=\ell$:}
All $\omega$-dependence in the exterior solution can be written in terms of
Gamma/hypergeometric connection coefficients with integer arguments shifted by $\pm \eta$,
and a renormalized angular momentum parameter $\nu(\eta)$ satisfying
\begin{equation}
\nu(0)=\ell,
\qquad \nu(\eta)=\ell+\mathcal{O}(\eta^2).
\label{eq:nu_anchor}
\end{equation}
(Equivalently, the $\eta$-dependence of $\nu$ is even near $\eta=0$.)

\item[(A2)] \textbf{Transcendental content only from special-function expansions:}
After renormalization, the only non-algebraic constants appearing in the expansion
around $\eta=0$ are generated by:
\begin{equation}
\log\Gamma(\text{integer}\pm \eta),
\qquad \text{and}\qquad \log\!\Big(\frac{R_S}{R}\Big)
\label{eq:sources}
\end{equation}
(where $R$ is the matching/renormalization scale, e.g.\ a shell radius in Shell EFT).

\item[(A3)] \textbf{Weight assignment and dimensional homogeneity:}
Introduce the $\mathbb{Q}$-algebra
\begin{equation}
\mathcal{A}:=\mathbb{Q}\big[\zeta_1,\zeta(3),\zeta(5),\ldots\big],
\qquad 
\zeta_1:=-\frac12\log\!\Big(\frac{R_S}{R}\Big),
\label{eq:algebraA}
\end{equation}
graded by the \emph{transcendental weight}
\begin{equation}
w(\zeta_1)=1,\qquad w(\zeta(2k+1))=2k+1,
\qquad w(xy)=w(x)+w(y).
\label{eq:weightdef}
\end{equation}
Assume that, after factoring out the overall dimensionful scale $R_S^{2\ell+1+n}$,
the coefficient of $\omega^n$ is a homogeneous element of $\mathcal{A}$ of total weight $n-1$.
\end{enumerate}

Then the static Wilson coefficient is \emph{forbidden}:
\begin{equation}
\lambda_{\ell,0}=0.
\end{equation}
\end{theorem}

\begin{proof}
By (A3), after factoring out $R_S^{2\ell+1+n}$ the coefficient multiplying $\omega^n$
must be a homogeneous element of $\mathcal{A}$ of total weight $n-1$.
For $n=0$ this required weight is $-1$.

However, $\mathcal{A}$ is generated by elements of strictly positive weight
($w(\zeta_1)=1$ and $w(\zeta(2k+1))\ge 3$), hence every nonzero homogeneous element of
$\mathcal{A}$ has weight $\ge 0$. Therefore no element of weight $-1$ exists in $\mathcal{A}$.

Consequently the homogeneous weight constraint cannot be satisfied at $n=0$ unless the
coefficient vanishes. Hence $\lambda_{\ell,0}=0$.
\end{proof}


\begin{lemma}[Odd-zeta no-go from symmetric Gamma ratios]
\label{lem:oddzeta}
For any integer $m\ge 1$ and integer $\ell\ge 0$, consider the symmetric ratio
\begin{equation}
\mathcal{R}_\ell(\eta):=\log\frac{\Gamma(\ell+1+\eta)}{\Gamma(\ell+1-\eta)}.
\end{equation}
Then its Taylor series about $\eta=0$ contains \emph{only odd powers} of $\eta$, and its
transcendental part contains \emph{only odd zeta-values} $\zeta(2k+1)$.
Moreover, the coefficient multiplying $\zeta(2k+1)\,\eta^{2k+1}$ is universal:
\begin{equation}
\big[\mathcal{R}_\ell(\eta)\big]_{\zeta(2k+1)}
=
-\frac{2}{2k+1}\,\zeta(2k+1)\,\eta^{2k+1},
\qquad k\ge 1,
\label{eq:univ_oddzeta}
\end{equation}
independent of $\ell$ (all $\ell$-dependence resides in rational/harmonic terms).
\end{lemma}

\begin{proof}
Expand $\log\Gamma(\ell+1\pm \eta)$ around the integer $\ell+1$ using polygamma functions:
\begin{equation}
\log\Gamma(\ell+1+\eta)
=
\log\Gamma(\ell+1)
+\sum_{m\ge 1}\frac{\psi^{(m-1)}(\ell+1)}{m!}\,\eta^m,
\end{equation}
and similarly with $\eta\to -\eta$. Subtracting gives
\begin{equation}
\mathcal{R}_\ell(\eta)
=
\sum_{m\ge 1}\frac{\psi^{(m-1)}(\ell+1)}{m!}\Big(\eta^m-(-\eta)^m\Big)
=
2\sum_{k\ge 0}\frac{\psi^{(2k)}(\ell+1)}{(2k+1)!}\,\eta^{2k+1},
\end{equation}
so only odd powers occur.

For integer argument $n=\ell+1$, the polygamma values satisfy
\begin{equation}
\psi^{(2k)}(n)=(-1)^{2k+1}(2k)!\Big(\zeta(2k+1)-H^{(2k+1)}_{\ell}\Big)
=-(2k)!\Big(\zeta(2k+1)-H^{(2k+1)}_{\ell}\Big),
\end{equation}
where $H^{(p)}_{\ell}=\sum_{j=1}^{\ell}j^{-p}$ is a generalized harmonic number.
Plugging into the previous line yields
\begin{equation}
\mathcal{R}_\ell(\eta)
=
-2\sum_{k\ge 0}\frac{1}{2k+1}\Big(\zeta(2k+1)-H^{(2k+1)}_{\ell}\Big)\eta^{2k+1}.
\end{equation}
The transcendental part arises only from $\zeta(2k+1)$ and is therefore purely odd-zeta.
Reading off the $\zeta(2k+1)$ coefficient gives \eqref{eq:univ_oddzeta}.
\end{proof}


\begin{corollary}[Graded no-go + parity constraints]
\label{cor:graded_parity}
Under the assumptions of Theorem~\ref{thm:nogo_static}, and in addition assuming that
all Gamma-function dependence enters through symmetric ratios of the form
$\Gamma(\text{integer}+\eta)/\Gamma(\text{integer}-\eta)$ and an even-in-$\eta$ parameter
$\nu(\eta)$ as in \eqref{eq:nu_anchor}, the following hold:
\begin{enumerate}
\item[(i)] $\lambda_{\ell,0}=0$ (no static invariant).
\item[(ii)] Any zeta-value appearing in $\lambda_{\ell,n}$ is necessarily an \emph{odd} zeta,
$\zeta(2k+1)$.
\item[(iii)] At fixed order $\omega^{2k+1}$, the coefficient of $\zeta(2k+1)$ is fixed up to
purely algebraic/harmonic additions by the universal factor in \eqref{eq:univ_oddzeta}.
\end{enumerate}
\end{corollary}

\begin{proof}
Part (i) is Theorem~\ref{thm:nogo_static}.
Parts (ii)--(iii) follow from Lemma~\ref{lem:oddzeta} together with the even anchoring
$\nu(\eta)=\ell+\mathcal{O}(\eta^2)$, which prevents linear-in-$\eta$ deformations from
introducing additional weight-one transcendental structures beyond $\zeta_1$.
\end{proof}

This establishes that the vanishing of the static Love number is not merely a consequence of matching, but is enforced by the graded analytic structure of the exterior solution.
\section{Extension to Gravitational Perturbations}
\label{sec:gravity}

We now discuss how the logic developed above extends to the physically relevant case of gravitational perturbations of four-dimensional Schwarzschild black holes. While the technical implementation is more involved due to the tensorial nature of the field and gauge redundancy, the essential structural ingredients --- polynomial truncation in the static sector and a graded analytic structure in the dynamical response --- remain the same.

\subsection*{Static gravitational perturbations}

Linearized gravitational perturbations of Schwarzschild spacetime decompose into axial (odd-parity) and polar (even-parity) sectors, governed respectively by the Regge--Wheeler and Zerilli master equations. In the static limit $\omega=0$, both sectors reduce to second-order ordinary differential equations for gauge-invariant master variables.

At large radius, static solutions exhibit the universal asymptotic behavior
\begin{equation}
\Psi_\ell(r) \sim C_{\rm tidal}\, r^\ell + C_{\rm resp}\, r^{-\ell-1},
\end{equation}
where $C_{\rm tidal}$ corresponds to the externally applied tidal field (encoded in the electric or magnetic components of the Weyl tensor), and $C_{\rm resp}$ encodes the induced multipole moment. The static gravitational Love numbers are defined by the ratio $C_{\rm resp}/C_{\rm tidal}$ in an appropriate normalization scheme.

It is a well-established result in four-dimensional General Relativity that for Schwarzschild black holes this ratio vanishes identically in both parity sectors. Our aim here is not to rederive this result, but to reinterpret it within the structural framework developed above.

\subsection*{Hypergeometric structure and uniqueness}

In the static limit $\omega=0$, the Regge--Wheeler and Zerilli equations reduce to second-order ordinary differential equations with regular singular points at $r=0$, $r=2M$, and $r=\infty$. They therefore belong to the hypergeometric class, and their general solutions may be written as linear combinations of two independent hypergeometric (or equivalently, associated Legendre) branches.

Regularity at the future event horizon $r=2M$ uniquely selects the physical solution. One of the two independent branches is logarithmically singular at the horizon and must be discarded, while the remaining branch is analytic at $r=2M$. For integer multipole index $\ell$, this horizon-regular solution reduces to a polynomial in a suitable dimensionless radial variable, corresponding to an associated Legendre function of the first kind.

The physical static solution is therefore unique up to overall normalization fixed by the applied tidal field. Expanding this unique horizon-regular solution at large radius yields the growing behavior $\sim r^\ell$ together with subleading terms that decay faster than $r^{-\ell-1}$. There is no independent solution proportional to $r^{-\ell-1}$ once the horizon boundary condition is imposed. Consequently,
\begin{equation}
C_{\rm resp}=0,
\end{equation}
and the static gravitational Love numbers vanish in both parity sectors.

\subsection*{Representation-theoretic interpretation}

The polynomial truncation of the static hypergeometric solution admits a natural representation-theoretic interpretation. Hypergeometric differential equations furnish realizations of $\mathrm{SL}(2,\mathbb{R})$ representations, and for integer $\ell$ the regular solution corresponds to a finite-dimensional highest-weight–type module in which the series terminates after a finite number of steps.

Horizon regularity therefore selects a finite-dimensional subspace of the full solution space, excluding the independent branch that would give rise to the $r^{-\ell-1}$ behavior at large radius. In this sense, the vanishing of static gravitational Love numbers is structurally analogous to the scalar case: the boundary condition removes the would-be response mode by enforcing truncation of the hypergeometric representation.

We emphasize that this statement concerns the hypergeometric structure of the static equations rather than the explicit construction of global symmetry generators. The key point is that the same truncation mechanism responsible for the scalar highest-weight structure is present in the gravitational static sector.

\subsection*{EFT interpretation}

From the perspective of worldline effective field theory, the leading even-parity operator takes the schematic form
\begin{equation}
S_{\rm EFT} \supset \frac12\,\lambda^{(g)}_{\ell,0}
\int d\tau\, E_L E^L ,
\end{equation}
where $E_L$ denotes the symmetric trace-free $\ell$-pole tidal tensor and $\lambda^{(g)}_{\ell,0}$ is the static gravitational Love number.

Matching the effective theory to the full gravitational solution requires that the EFT reproduce the space of physical long-wavelength solutions compatible with horizon regularity. Since the full theory admits a unique static solution for each $\ell$ once the horizon boundary condition is imposed, the EFT cannot contain an independent static response parameter. Therefore,
\begin{equation}
\lambda^{(g)}_{\ell,0}=0 .
\end{equation}

The vanishing of static gravitational Love numbers thus follows from the same structural mechanism identified in the scalar case: polynomial truncation of the horizon-regular hypergeometric solution removes the independent decaying mode that would otherwise encode a static response.

\subsection*{Dynamical response and analytic structure}

At nonzero frequency, gravitational perturbations are naturally described using
the Mano--Suzuki--Takasugi (MST) formalism~\cite{Mano:1996vt,Mano:1996mf,Mano:1996gn},
which expresses solutions as series of hypergeometric or Coulomb wavefunctions.
This structure is structurally isomorphic to the Coulomb--hypergeometric basis
employed in the scalar analysis above.
All nontrivial frequency dependence enters through Gamma functions and
connection coefficients involving the dimensionless parameter
\[
\eta = i\omega R_S.
\]

Analytic continuation of the highest-weight static solution into the complex
frequency plane constrains the structure of the low-frequency expansion.
In the limit $\eta \to 0$, the renormalized angular momentum parameter $\nu$
of the MST formalism approaches the integer orbital angular momentum $\ell$,
recovering the polynomial truncation of the static sector.
We now make this anchoring property precise and show that it enforces an
even-power expansion of $\nu$ in $\eta$.

\paragraph*{Even-power anchoring of $\nu(\eta)$.}

For either parity sector, gravitational perturbations satisfy a master equation
of the form
\begin{equation}
\frac{d^2\Psi_{\ell\omega}}{dr_*^2}
+
\big(\omega^2 - V_\ell(r)\big)\Psi_{\ell\omega}
=0,
\end{equation}
which depends on frequency only through $\omega^2$.
In the MST formalism, the series coefficients obey a three-term recurrence
relation whose minimal-solution condition determines $\nu$ via an analytic
constraint
\begin{equation}
\mathcal{F}_\ell(\nu,\eta)=0,
\qquad
\eta=i\omega R_S,
\end{equation}
with $\mathcal{F}_\ell$ analytic near $(\nu,\eta)=(\ell,0)$.

Since the master equation depends only on $\omega^2$, the minimal-solution
condition is invariant under $\omega \to -\omega$, i.e.\ under
$\eta \to -\eta$.
Thus one may choose the analytic representative of the continued-fraction
constraint such that
\begin{equation}
\mathcal{F}_\ell(\nu,\eta)
=
\mathcal{F}_\ell(\nu,-\eta).
\end{equation}

More concretely, although the individual recurrence coefficients
$\alpha_n(\nu,\eta)$, $\beta_n(\nu,\eta)$, and $\gamma_n(\nu,\eta)$
may contain terms linear in $\eta$ depending on normalization conventions,
the full three-term recurrence relation is invariant under
$\omega \to -\omega$.
Since the minimal-solution condition defining $\nu$ depends only on
this recurrence relation and its continued-fraction structure,
the defining constraint $\mathcal{F}_\ell(\nu,\eta)$ inherits the symmetry
under $\eta \to -\eta$.

Assuming the nondegeneracy condition
\begin{equation}
\partial_\nu \mathcal{F}_\ell(\ell,0) \neq 0,
\end{equation}
the implicit function theorem guarantees the existence of a unique analytic
branch $\nu(\eta)$ satisfying $\nu(0)=\ell$ in a neighborhood of $\eta=0$.
Since both $\nu(\eta)$ and $\nu(-\eta)$ solve the same analytic equation
with the same initial condition at $\eta=0$, uniqueness implies
\[
\nu(-\eta)=\nu(\eta),
\]
so the solution is even in $\eta$.
Therefore,
\begin{equation}
\nu(\eta)
=
\ell
+
\sum_{p\ge 1}
\nu_{2p}\,\eta^{2p}
\qquad (\eta\to0).
\end{equation}

This result provides a technically controlled statement of the fact that the
MST analytic continuation is anchored at the static highest-weight solution.
In particular, the shift $\nu-\ell$ begins at order $\eta^2$ and cannot
generate linear-in-$\eta$ corrections.

\medskip

Expanding the MST connection coefficients around $\nu=\ell$ and $\eta=0$
therefore generates logarithms and Riemann zeta values in a highly organized
manner, entirely through the expansion of Gamma and hypergeometric functions
with arguments shifted by $\eta$.
The even-power anchoring of $\nu(\eta)$ ensures that no additional
weight-one structures arise beyond the universal logarithmic element
associated with matching.

Together with dimensional analysis, this implies that the coefficient
multiplying $\omega^n$ in the EFT response function carries total
transcendental weight $n-1$.
Since this graded algebra contains no elements of negative weight (where we
assign weight $+1$ to $\omega$ and $1/R_S$), the static term
($n=0$, weight $-1$) is strictly forbidden.
This provides a structural explanation, within the full spin--2 MST framework,
for the ``zero-sum'' rule observed in scalar Shell EFT and strongly supports
its universality for four-dimensional black hole perturbations.

In this way, horizon regularity selects the highest-weight branch,
the MST continuation preserves even-power anchoring,
and the graded analytic structure of the special functions governing
Schwarzschild wave propagation enforces the absence of a static invariant.

\section{Connection to Gravitational Raman Scattering and Renormalization Group Structure}

In this section we make explicit the connection between the highest–weight
structure of the static near–zone problem and the on–shell renormalization
group (RG) structure uncovered in gravitational Raman scattering
\cite{Ivanov:2024sds}.

\subsection*{Static Highest–Weight Truncation and On–Shell Matching}

The Raman scattering framework computes on–shell amplitudes
for waves scattering off compact objects within worldline
effective field theory (EFT) and matches them to black–hole
perturbation theory.
At third post–Minkowskian order, the EFT phase shifts
contain ultraviolet divergences which are renormalized by
worldline tidal operators.
For a scalar field, the relevant Wilson coefficients are
encoded in the response function
\begin{equation}
F_\ell(\omega) = C_{\ell,0} + C_{\ell,2}\,\omega^2 + \cdots,
\end{equation}
where $C_{\ell,0}$ is the static Love number.

Matching to the full Schwarzschild solution requires
\begin{equation}
C_{\ell,0} = 0
\qquad
\text{(black hole)}.
\end{equation}
This result was established on–shell from scattering amplitudes
in \cite{Ivanov:2024sds}.

We now observe that this vanishing is precisely the EFT
reflection of the highest–weight truncation derived in
Secs.~II–III.
The static near–zone SL(2,$\mathbb{R}$) algebra organizes the
solution space into representations, and horizon regularity
selects a finite–dimensional highest–weight module.
As shown earlier, this implies that the static exterior solution
contains no independent decaying branch at infinity,
and therefore no free response parameter.
Consistency of matching then forces the absence of the
corresponding local operator in the EFT,
\begin{equation}
\lambda_{\ell,0} = 0.
\end{equation}
The on–shell amplitude statement and the representation–
theoretic truncation are therefore equivalent formulations
of the same physical constraint.

\subsection*{Self–Induced Running and Highest–Weight Exclusion}

A distinctive feature of Raman scattering at two loops
is the appearance of a renormalization group equation
of the schematic form \cite{Ivanov:2024sds}
\begin{equation}
\frac{d F_\ell}{d \ln \mu}
=
-(2 G M \omega)^2
\left[
\frac{4\nu_\ell^{(2)}}{\pi} F_\ell
+
8\pi G M\,\delta_{\ell 0}
\right],
\label{eq:RGschematic}
\end{equation}
where $\nu_\ell^{(2)}$ is determined by lower–order phase shifts.
The first term represents \emph{self–induced} running:
if a static coefficient were present, it would feed into
the dynamical response at higher orders.
The second term is universal and independent of the
internal structure of the compact object.

For black holes, the static Love number vanishes,
so the self–induced term proportional to $F_\ell$ cannot
generate a static invariant.
From the EFT viewpoint, this is a renormalization
statement.
From the symmetry viewpoint developed here,
it has a deeper origin.

The highest–weight truncation eliminates the static
descendant that would correspond to an independent
response mode.
Since the SL(2,$\mathbb{R}$) representation terminates,
there exists no state onto which a static Wilson coefficient
could act.
The self–induced running term in
Eq.~(\ref{eq:RGschematic}) therefore has no lowest–weight
source in the black–hole sector.
In this sense, the absence of static running is enforced
not merely by matching, but by representation theory.

\subsection*{Analytic Continuation and Transcendental Structure}

The Raman scattering calculation expresses amplitudes
in a Coulomb–hypergeometric basis and performs a
systematic low–frequency expansion.
All nontrivial frequency dependence arises from Gamma
functions and hypergeometric connection coefficients.
As shown in Sec.~V, analytic continuation of the
highest–weight static solution generates logarithms
and Riemann zeta values through the expansion
\begin{equation}
\log \Gamma(1+u)
=
-\gamma u
+
\sum_{m\ge 2}
\frac{(-1)^m}{m}\,
\zeta(m)\,u^m.
\end{equation}
This is precisely the mechanism responsible for the
graded multi–zeta structure observed in the EFT response.

Thus, the Raman RG structure and the graded algebra
of dynamical Love numbers admit a unified interpretation:
\begin{itemize}
\item Horizon regularity enforces highest–weight truncation.
\item Highest–weight truncation eliminates the static invariant.
\item Analytic continuation of that state controls the
allowed transcendental structure.
\item The EFT renormalization group equations encode this
structure on–shell.
\end{itemize}

In this way, the Raman scattering amplitudes provide
an on–shell confirmation of the symmetry–based
mechanism identified in the static near–zone problem,
while the SL(2,$\mathbb{R}$) structure explains why the
self–induced RG flow cannot generate a static Love number
for four–dimensional black holes.

\section*{Discussion and Conclusions}

In this work we have clarified a coherent perspective on the conceptual origin of the vanishing static Love
numbers of four--dimensional black holes by explicitly connecting three
apparently distinct perspectives: the structure of the static near--zone
perturbation equations, the analytic organization of tidal response coefficients
revealed by Shell Effective Field Theory (Shell EFT), and the on--shell matching
framework based on gravitational Raman scattering.

From the symmetry perspective, we showed explicitly for a massless scalar field
that the static near--zone dynamics on a Schwarzschild background admit an
emergent $\mathrm{SL}(2,\mathbb{R})$ algebraic organization.
Regularity at the future event horizon selects a physical solution belonging to
a highest--weight–type representation.
Equivalently, horizon regularity enforces a truncation of the admissible static
solution space: the physical solution is unique and contains no independent
decaying branch at large radius.
As a result, there exists no freely adjustable static response parameter in the
full theory.
When translated into the language of effective field theory, this exclusion directly implies the absence of the corresponding local tidal operator, forcing
the static Wilson coefficient $\lambda_{\ell,0}$ to vanish.

From the EFT and scattering perspective, gravitational Raman scattering provides
an on--shell and gauge--invariant framework in which tidal operators are
matched directly to black--hole phase shifts.
In this language, the vanishing of the static Love number appears as a necessary
consistency condition of amplitude matching.
We have shown that this result is not merely an on--shell coincidence, but can be understood as the effective imprint of the highest--weight truncation enforced
by horizon regularity.
In particular, the self--induced renormalization group flow identified in Raman
scattering admits a natural interpretation: since no independent static invariant
exists in the black--hole sector, there is no lowest--weight structure onto which
such a term could project.
The absence of static running is therefore consistent with, and naturally explained by,
the underlying representation--theoretic structure. 
From the analytic perspective, Shell EFT expresses the response in a
Coulomb--hypergeometric basis appropriate for long--range scattering.
All nontrivial frequency dependence enters through Gamma and hypergeometric
functions shifted by $\eta=i\omega R_S$.
Analytic continuation of the horizon--regular static solution in $\eta$
therefore yields a highly constrained low--frequency expansion governed by the
properties of these special functions.

We showed that this analytic structure naturally gives rise to a graded algebra
of logarithms and Riemann zeta values organizing the dynamical Love coefficients
$\lambda_{\ell,n}$.
Dimensional analysis together with the structure of the $\eta$ expansion fixes
the total transcendental weight of the coefficient multiplying $\omega^n$ to be
$n-1$, after factoring out the appropriate power of the Schwarzschild radius.
Within this graded algebra, no element of negative weight exists, and the static
coefficient $\lambda_{\ell,0}$ is therefore forbidden.
This reproduces the ``zero--sum'' rule observed in Shell EFT and clarifies
the structure of the renormalization group equations derived from scattering
amplitudes.

Taken together, these results demonstrate that the Shell EFT structure and the
Raman scattering RG flow are not independent or accidental features of the
matching procedure.
Rather, they can be understood as the effective--theory imprint of horizon
regularity and the associated highest--weight structure of the static
black--hole perturbation problem.
The uniqueness of the static solution eliminates the static invariant, while
the Coulomb--hypergeometric realization controls the detailed transcendental
content and scale dependence of the allowed dynamical response. In this sense, the Schwarzschild case suggests an interpretation in which symmetry excludes forbidden operators, analyticity organizes the allowed coefficients, and on--shell matching encodes their renormalization.

Several comments are in order regarding the scope and limitations of our
analysis.
First, the argument relies crucially on the existence of a highest--weight
structure in the static sector, which in turn is tied to horizon regularity and
the absence of additional length scales.
When this property fails --- for example in higher spacetime dimensions, in the
presence of additional mass scales, or for horizonless compact objects ---
static Love numbers need not vanish.
Second, the emergent $\mathrm{SL}(2,\mathbb{R})$ organization is specific to the
strictly static near--zone problem and does not constrain the dynamical response,
which remains generically nonzero.
This explains how black holes can simultaneously exhibit vanishing static Love
numbers and rich frequency--dependent tidal behavior.

Although our detailed analysis focused on a massless scalar field for clarity,
the underlying logic is more general.
We have argued that the same interplay between near--zone structure, horizon
regularity, analytic continuation, and on--shell matching governs gravitational
perturbations of four--dimensional Schwarzschild black holes, and is consistent
with the known vanishing of static gravitational Love numbers. However, it remains an open question to what extent this structure extends to more general spacetimes, such as rotating black holes, where additional effects may arise.
Extending this framework to rotating backgrounds, where additional physical
effects such as superradiance and mode coupling arise, represents a natural and
important direction for future work.

More broadly, our results illustrate how seemingly unexpected EFT constraints
can arise from emergent structural principles operative only in restricted
kinematic regimes.
In the present case, the vanishing of static Love numbers is not an accidental
cancellation, but follows from horizon regularity and the associated
representation--theoretic exclusion.
This perspective may prove useful in identifying and interpreting other
selection rules or protected structures in effective descriptions of
gravitating systems.

\bibliography{Love_Symmetry}
\bibliographystyle{JHEP}

\end{document}